\pdfoutput=1

\documentclass[runningheads]{llncs}
\usepackage{amssymb}
\usepackage{amsmath}
\usepackage{graphicx}
\usepackage{verbatim}
\usepackage{color}
\usepackage{ae}
\usepackage[ruled]{algorithm}
\usepackage{algorithmicx}
\usepackage{algpseudocode}
\usepackage[breaklinks,bookmarks=false]{hyperref}

\DeclareMathOperator{\post}{post}
\DeclareMathOperator{\pre}{pre}

\DeclareMathOperator{\ancestors}{path}
\DeclareMathOperator{\weight}{weight}

\DeclareMathOperator{\depth}{depth}

\begin{document}

\title{(Really) Tight bounds for dispatching binary methods}
\titlerunning{Dispatching binary methods}

\author{Pawe\l{} Gawrychowski\thanks{Supported by MNiSW grant number N~N206 492638, 2010--2012 and START scholarship from FNP.}}
\institute{
     Institute of Computer Science, University of Wroc{\l}aw, Poland\\
     Max-Planck-Institute f\"{u}r Informatik, Saarbr\"ucken, Germany\\
	\email{gawry@cs.uni.wroc.pl}
	}

\maketitle
\begin{abstract}
We consider binary dispatching problem originating from object oriented programming. We want to preprocess a hierarchy of classes and collection of methods so that given a function call in the run-time we are able to retrieve the most specialized implementation which can be invoked with the actual types of the arguments. This problem has been thoroughly studied for the case of mono dispatching~\cite{Muthukrishnan,FerraginaDynamic}, where the methods take just one argument, resulting in (expected) $\mathcal{O}(\log\log m)$ query time after just linear preprocessing. For the binary dispatching, where the methods take exactly two arguments, logarithmic query time is possible~\cite{Ferragina}, even if the structure is allowed to take linear space~\cite{Alstrup}. Unfortunately, constructing such structure requires as much as (expected) $\Theta(m(\log\log m)^2)$ time~\cite{Alstrup,Poon}.

Using a different idea we are able to construct in (deterministic) linear time and space a structure allowing dispatching binary methods in the same logarithmic time. Then we show how to improve the query time to just $\mathcal{O}(\frac{\log m}{\log\log m})$, which is easily seen to be optimal as a consequence of some already known lower bounds if we want to keep the size of the resulting structure close to linear.

\end{abstract}

\section{Introduction}

The motivation for the method dispatching comes from object-oriented programming languages, where we have a hierarchy of classes with uniquely defined parents. We also have a collection of $m$ functions accepting a constant number of arguments, where each argument must have a specified class as an ancestor in the hierarchy. Then given a function call, we should (efficiently) determine the most specific implementation based on the actual types of the arguments. This problem was first considered in the mono dispatching version, where each method takes just one argument. It is known that in such case the input can be preprocessed in linear time and space so that each query can be answered in $\mathcal{O}(\log\log m)$ time~\cite{Muthukrishnan}, and that we can update the structure in $\mathcal{O}(m^\epsilon)$ time while retaining the same bounds on the query time~\cite{Ferragina}.

The more general multi-method version of the problem was considered by Ferragina \emph{et al.}~\cite{Ferragina}, whose methods imply that for the special case of the binary dispatching (where all function are binary) we can achieve $\mathcal{O}(\log\log m)$ query time after a $\mathcal{O}(m^{1+\epsilon})$ preprocessing or $\mathcal{O}(\log m)$ query after a $\mathcal{O}(m\log m)$ preprocessing. Then Eppstein and Muthukrishnan~\cite{Eppstein} improved the former by showing that, for example, $\mathcal{O}(1)$ query time is possible after a $\mathcal{O}(m^{1+\epsilon})$ preprocessing. Finally Alstrup \emph{et al.}~\cite{Alstrup} decreased the preprocessing space to $\mathcal{O}(m)$ while retaining logarithmic query time. Unfortunately, their preprocessing time was not linear but (expected) $\mathcal{O}(m(\log\log m)^{2})$. Another structure with the same bounds was given by Poon and Kwok~\cite{Poon}. 

In this paper we first give in Section~\ref{section:algorithm} a simpler yet more effective solution with (deterministic) linear time and space preprocessing and the same logarithmic query time. Our solution uses a slightly different approach than the previous methods, and because of this difference we are then able to decrease the query time in Section~\ref{section:algorithm2} to just $\mathcal{O}(\frac{\log m}{\log\log m})$ while retaining the linear time and space preprocessing. This complexity is easily seen to be optimal for structures of size $\mathcal{O}(m\log^{c}m)$ as a consequence of some already known lower bounds, which we briefly review in Section~\ref{section:lower bound}.

While even the first logarithmic query time solution needs the word RAM model, the same is true for the previously known linear space solutions, hence it should not be seen as a drawback.

\begin{figure}[t]
\centering
\includegraphics{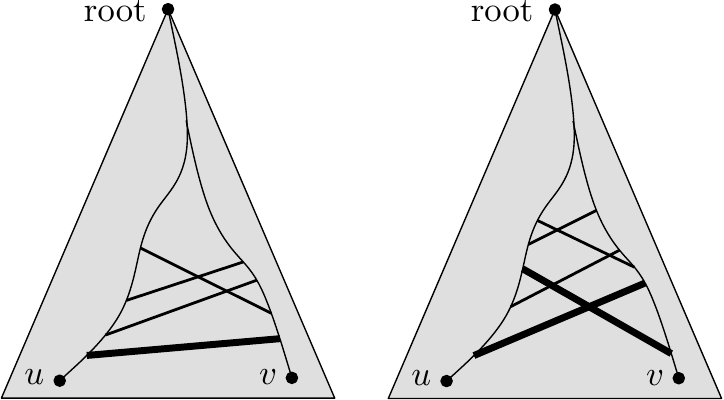}
\caption{Unique lowest bridge for $u$ and $v$ and an ambiguous situation.}
\label{figure:bridges}
\end{figure}

\section{Preliminaries}

We work with the following formulation of the problem: we are given a tree $T$ on $n$ vertices and a collection of $m$ \emph{bridges}, which are simply pairs of vertices $(u,v)$. We say that a bridge $(u,v)$ is lower than $(u',v')$ if $u'$ is a descendant of $u$ and $v'$ is a descendant of $v$. We want to preprocess the input so that given two vertices $u'$ and $v'$ we can detect the lowest bridge $(u,v)$ such that $u'$ is a descendant of $u$ and $v'$ is a descendant of $v$, see Figure~\ref{figure:bridges}. If there is no such unique lowest bridge, we need to signal ambiguity. We aim to develop a $\mathcal{O}(n+m)$ time preprocessing which allows $\mathcal{O}(\log m)$ time queries.

We work in the standard RAM model of computation with logarithmic word size. The following result is known in such model.

\begin{lemma}[atomic heaps~\cite{FredmanWillard}]\label{lemma:atomic}
It is possible to maintain a collection of sets $S(i)\subseteq\{1,2,\ldots,n\}$ so that inserting, removing and finding successor in each of those sets work in constant time (amortized for insert and remove, worst case for find) as long as $|S(i)|\leq\log^{c} n$ for all $i$, assuming $\mathcal{O}(n)$ time and space preprocessing, where $c$ is any (but fixed) constant.
\end{lemma} 

\section{New algorithm}
\label{section:algorithm}

First observe that the above problem reduces in a natural way (by computing the pre- and post-order numbers) to retrieving the smallest rectangle containing a given point on a $n\times n$ grid (or detecting there is no such unique smallest rectangle). From now on we will work with this simple geometric formulation. Note that the $x$ and $y$ projections of any two rectangles are either disjoint or contained in each other (we call such collection of rectangles {\it valid}) and we can normalize the coordinates so that $n\leq 2m$.

We sweep the grid from left to right while maintaining a structure describing currently intersected rectangles. The structure is simply a full binary tree on $n$ leaves corresponding to different $y$ coordinates. To process an interval $[y_1,y_2]$ with
$y_1<y_2$, we locate the lowest common ancestor $v$ of the leaves corresponding to $y_1$ and $y_2$ and call it {\it responsible} for $[y_{1},y_{2}]$. Each inner vertex stores a stack containing all intervals it is currently responsible for. To insert a new interval we push it onto its responsible vertex stack. To remove an interval, locate the responsible vertex and observe that (because the collection is valid) the interval we want to remove is its top element, and we can simply pop it.

\begin{figure}[t]
\centering
\includegraphics[scale=0.9]{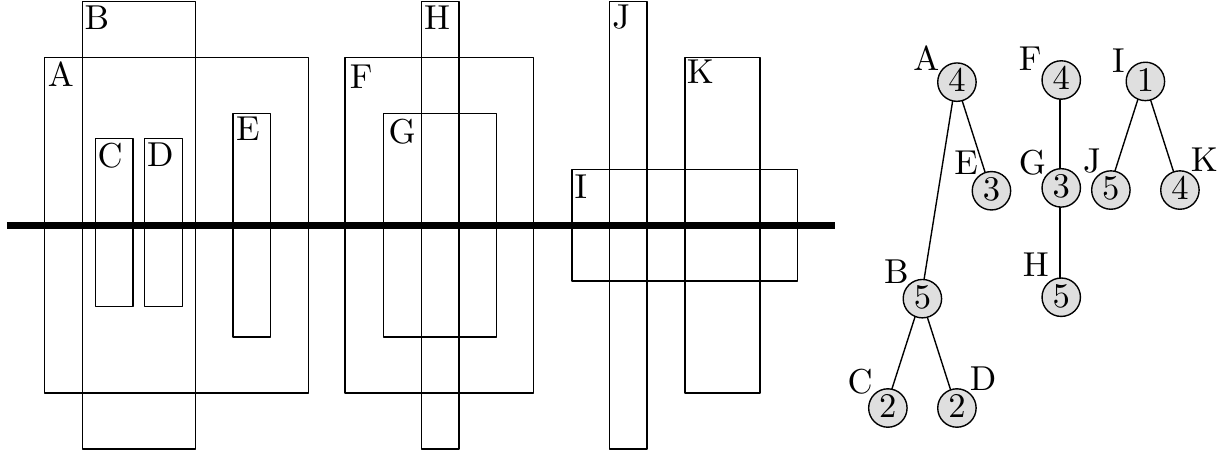}
\caption{Rectangles on the same stack and their version tree.}
\label{figure:responsible}
\end{figure}

Fix an inner vertex and consider all $\ell$ rectangles it was responsible for, see Figure~\ref{figure:responsible}. Note that the intersection of
all their $y$ projections is nonempty, and hence any two of those projections are contained in each other. By a simple linear time transformation we can assume that all their start and end points are different, and their sorted list is $x_1<x_2<\ldots<x_\ell$ (of course we cannot assume that sorting a single list can be performed in linear time, but note that we can sort lists of all inner vertices at once, and because their elements are small integers we can apply counting sort). We define the \emph{version tree} as follows: the parent of a rectangle $[x_1,x_2]\times [y_1,y_2]$ is the rectangle $[x'_1,x'_2]\times [y'_1,y'_2]$ such that $[x_1,x_2]\subseteq [x'_1,x'_2]$ and $[x'_1,x'_2]$ is the smallest possible. Because any two $x$ projections are either disjoint or contained in each other, and all $x_i$ are different, this is a valid definition. We
may assume that the result is indeed a tree (not a forest) by adding one artificial rectangle. Each vertex of this tree is labeled with an integer denoting the height $y_2-y_1$ of the corresponding rectangle. Consider the sorted list of all different $x$ coordinates. For each pair of consecutive
integers $x_i<x_{i+1}$ on this list we would like to find the vertex of the version tree such that its ancestors are exactly the elements of the stack at time $t\in(x_i,x_{i+1})$ (we call it the \emph{tail} at time $t$). This can be precomputed in a straightforward way during the sweep. 

Consider a query concerning a point $(x,y)$. First we locate the leaf $v$ corresponding to $x$. Any rectangle containing $(x,y)$ belongs to the
stack of one of its ancestors at time $x$. More specifically, it must be an ancestor of the tail at time $x$ of one of those $\log m$ stacks. Hence
we should start with locating all those $\log m$ tails efficiently.

\begin{lemma}\label{lemma:cascading}
Given a time $t$ we can retrieve its tail at every ancestor of the leaf corresponding to $v$ in total $\mathcal{O}(\log m)$ time after a linear time and space preprocessing.
\end{lemma}

\begin{proof}
A straightforward application of the fractional cascading technique of Chazelle~\cite{Chazelle}. Recall that this technique allows linear time and
space preprocessing of a constant-degree graph with (sorted) lists of elements associated to the vertices so that given a path we can perform binary search for the same value in all lists corresponding to its vertices in time $\mathcal{O}(\log m + p)$, where $m$ is the total size of all lists and $p$ is the length of the path. In our case $p=\log m$ and the claimed running time follows.
\qed
\end{proof}

Each version tree will be carefully preprocessed as to implement two operations. Let $\ancestors(v)$ be the set of all ancestor weights of a given vertex $v$. The first operation is very simple: given $v$ we would like to find the vertex corresponding to $\max\ancestors(v)$. This can be trivially preprocessed in linear time and space. The second operation is more involved: given $v$ and $x$ we would like to find the vertex corresponding to the successor of $x$ in $\ancestors(v)$. Before we show how to implement it efficiently, we formulate two auxiliary lemmas.

\begin{figure}
\centering
\includegraphics[scale=0.9]{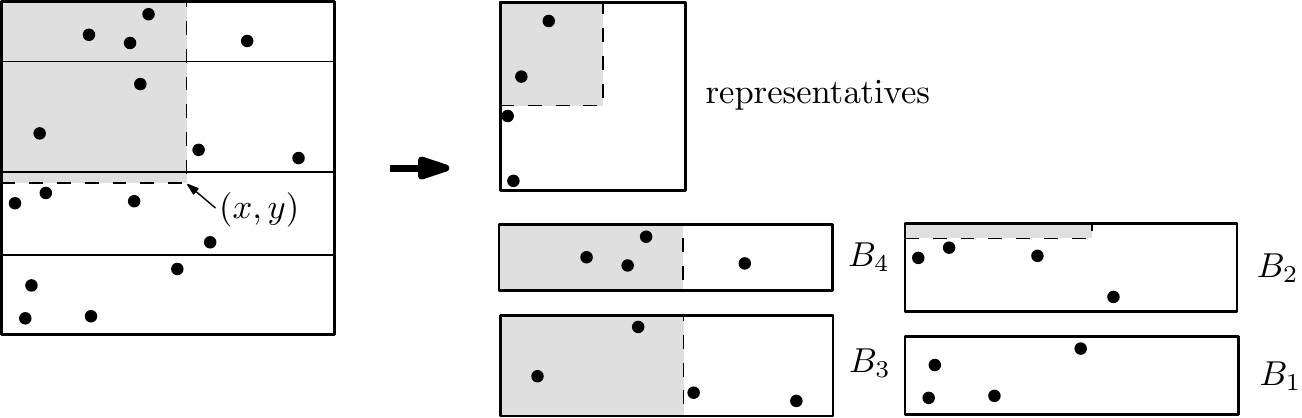}
\caption{Reducing the original query to queries in smaller structures.}
\label{figure:slices}
\end{figure}

\begin{lemma}\label{lemma:small grid}
A collection of sets of points $S(i)$ on a $n\times n$ grid such that $|S(i)|\leq\log^{2} n$ for all $i$ can be preprocessed in $\mathcal{O}(n+\sum_{i}|S_{i}|)$ time so that given $i$ and $(x,y)$ we can retrieve the point corresponding to $\min\{y' : (x',y')\in S(i)\wedge x'\leq x\wedge y'\geq y \}$ in constant time.
\end{lemma}

\begin{proof}
The idea is to recursively build a collection of smaller structures allowing performing the same operation but on subsets of the original set of point. 

Assume we have a set of $m \leq \log^2 n$ points $(x_i,y_i)$. By Lemma~\ref{lemma:atomic} we can sort the points according to their $y$ coordinates in $\mathcal{O}(m)$ time. Then we split the original set into blocks of size roughly $\sqrt{m}$ by choosing $\sqrt{m}$ evenly spaced elements $(x_{i\sqrt{m}},y_{i\sqrt{m}})$ in this sorted sequence and call the $i$-th block $B_i$. We build a smaller structure for each $B_i$. Additionally, let $x'_i$ be the smallest $x$ coordinate in the $i$-th block and create a new set of $\sqrt{m}$ points of the form $(x'_i,i)$ which we call \emph{representatives}. We build a smaller structure for this new set. If $m$ is very small, say $m\leq\sqrt{\log n}$, we switch to a different method: we can first normalize the coordinates of the points so that they are at most $\sqrt{\log m}$ and encode the whole set in a single machine word going row-by-row.

Observe that the total size of all structures built for a single $S(i)$ is just $\mathcal{O}(S(i))$. Furthermore, they allow us to answer a single query in constant time as follows. First locate the block $y$ belongs to and query the corresponding smaller structure. If this smaller structure contains a point $(x',y')$ with $x'\leq x$ and $y'\geq y$, we are done. Otherwise we use the smaller structure built for the representatives to locate the lowermost block containing such point and query its corresponding smaller structure, see Figure~\ref{figure:slices}. If the size of $S(i)$ is small and we have the whole set encoded in a single machine word, we can answer a query by first masking out all bits corresponding to points with too big $x$ or too small $y$ coordinates and then finding the lowest bit set to $1$. The total running time is constant because we will inspect just a constant number of structures for a single query.
\qed
\end{proof}

The next lemma will be used to preprocess each version tree. The main tool in its proof is the heavy path decomposition, which is defined as follows: each vertex chooses an edge leading to a child with the largest size. Removing all non-chosen edges leaves us with a collection of paths. We define the \emph{path tree} by creating one vertex for each such path, and choosing the parent of a path $p$ by looking up the parent of its highest vertex in the original tree and retrieving the corresponding path. It is easy to see that the depth of a path tree is just $\log n$. We say that a path $p$ is above a vertex $v$ if the path $v$ belongs to is a descendant of $p$ in the path tree. 

\begin{lemma}\label{lemma:version preprocessing}
A node weighted weighted tree on $n$ vertices with the weights from $\{1,2,\ldots,n\}$ can be preprocessed in linear time and space so that given $v$ and $x$ we can find the vertex corresponding to the successor of $x$ in $\ancestors(v)$ in $\mathcal{O}(\log n)$ time.
\end{lemma}

\begin{proof}
Consider a single path. By preprocessing all paths at once we can construct a sorted list of all weights $\weight(v_{1})\leq\weight(v_{2})\leq\ldots\leq\weight(v_{\ell})$ on this path. We split each list by choosing $\frac{\ell}{\log^{2} n}$ evenly spaced weights $\weight(v_{\alpha\log^{2} n})$ and call the corresponding vertices \emph{important}. The $\alpha$-th group contains vertices $v_{\alpha \log^{2}n-\Delta}$ for $\Delta=0,1,\ldots,\log^{2}n-1$. We choose the highest vertex from each such group and calling it the {\it representative}. Note that the total number of both important vertices and representatives is at most $\frac{n}{\log^{2} n}$.  For each group we construct the corresponding small set:
$$\left(\depth(v_{\alpha\log^{2}n-\Delta}),\weight(v_{\alpha\log^{2}n-\Delta})\right) \qquad\text{ for all } \Delta=0,1,\dots,\log^{2}n-1$$
and apply the preprocessing described in Lemma~\ref{lemma:small grid}.

We would like to implement the following operations efficiently:
\begin{enumerate}
\item given $v$ and $x$, for each path above $v$ find the successor of $x$ among the weights of its important vertices,
\item given $v$ and $x$, for each path above $v$ find the successor of $x$ among the weights of all representatives which are ancestors of $v$.
\end{enumerate}
First lets see how such information allows us quick retrieval of the successor of $x$ in $\ancestors(v)$. Assuming that we know the successor of $x$ among all weights of the important vertices on $p$, we query the structure constructed for the corresponding group. If it contains at least one vertex above $v$, we are clearly done. Otherwise we know that the successor belongs to a higher group than $x$. Assuming that we know the successor of $x$ among the weights of all representatives which are above $v$, we query the structure constructed for his group.

We use almost the same to implement both operations. Lets start with the former. For each path $p$ we build a binary search tree containing all important vertices located on all paths corresponding to the ancestors of $p$ (including $p$ itself) in the path tree. The vertices are sorted according to their weight. By using any persistent balanced search trees we can build the structures in total linear time and space. Furthermore, the total number of new nodes created as a result of all inserts will be just $\mathcal{O}(\frac{n}{\log n})$. For the sake of concreteness, assume that we use trees in which the original elements are stored only in the leaves. To facilitate efficient query processing, at each vertex $v$ we store an additional \emph{helper structure} mapping a path depth to the smallest weight stored at the subtree of $v$ and originating from an important vertex with such path depth. This structure consists of an array of size $\log n$ and one word with the $i$-th bit set if and only if the $i$-th entry in the array is defined. The helper structures are of just $\mathcal{O}(\log n)$ size, hence we can afford to build one for each new node in total linear time and space. Now given a query concerning a vertex $v$, we locate its path and the corresponding binary search tree. Then we find the successor of $x$ in this tree with a single $\mathcal{O}(\log n)$ time transversal. We claim that the helper structures stored at all right brothers of the visited vertices give us enough information to locate the successors of $x$ among the chosen weights of all paths above $v$. This is fairly obvious if we consider a single such path. The tricky part is to extract the information from all of them in $\mathcal{O}(\log n)$ time. We go through the vertices in a bottom-up order. In the very beginning we do not have the successor on any path. We iteratively consider the right brother of the next visited vertex: its helper structure gives us the successor on every path for which the corresponding entry is defined and which is yet unknown. By storing the currently unknown set in a single word, we can compute this intersection in constant time, and then extract the successors in constant time per path.

To implement the latter, we use almost the same method. The only exception is that now we build a binary search tree for each representative instead of each path.
\qed
\end{proof}

Now we are ready to prove the main lemma in this section.

\begin{lemma}\label{lemma:smallest}
A set of $m$ rectangles on a $n\times n$ grid can be preprocessed in linear time and space so that given a point $(x,y)$ we can find the rectangle $[x_{1},x_{2}]\times[y_{1},y_{2}]$ containing $(x,y)$ with the smallest height $y_{2}-y_{1}$ in $\mathcal{O}(\log m)$ time.
\end{lemma}

\begin{proof}
First apply Lemma~\ref{lemma:cascading} to locate the tail at time $x$ in every ancestor of the leaf corresponding to $x$. Then select $v$ to be the lowest of those ancestors such that the maximum on the corresponding tail-to-root path is sufficiently big to contain $y$ in the corresponding interval. There are just $\log m$ ancestors and extracting each maximum requires constant time. Now we claim that the smallest height rectangle can be found by looking at the version tree of $v$ only.  Assume otherwise, i.e., there is $v'$ such that $v'$ is an ancestor of $v$ and the smallest rectangle belongs to the version tree of $v'$. But then the $y$ interval at $v'$ properly contains the the $y$ interval at $v$, and hence the interval of the rectangle at $v$ is smaller.

To finish the proof, we apply Lemma~\ref{lemma:version preprocessing} to each version tree. Then given the tail at $v$ we first compute the smallest possible height of a rectangle stored in this tree which guarantees containing $y$ inside. This can be performed in $\mathcal{O}(\log m)$ time if we store all $y$ intervals sorted according to their lengths. Then we compute the successor of this smallest possible height on the tail-to-root path. It corresponds to the smallest height rectangle.
\qed
\end{proof}

\begin{theorem}
Bridge color problem can be solved in $\mathcal{O}(\log m)$ time after a linear time and space preprocessing.
\end{theorem}

\begin{proof}
First we transform the input so that no two bridges share an endpoint. This can be ensured by increasing the number of vertices to at most $n'=n+2m$ by repeating the following procedure: given a group of bridges $(u,v_{1})$, $(u,v_{2})$ ,$\ldots$, $(u,v_{k})$ with the same endpoint $u$, sort them so that $\depth(v_{i})<\depth(v_{i+1})$ for all $i=1,2,\ldots,k-1$. Then replace $u$ by a path $u_{1}\rightarrow u_{2}\rightarrow\ldots\rightarrow u_{k}$ and create $k$ bridges $(u_{i},v_{i})$ for $i=1,2,\ldots,k$.  We interpret the resulting set of bridges as a collection of rectangles on a $n'\times n'$ grid.

By Lemma~\ref{lemma:smallest} we can find the rectangle $[x_{1},x_{2}]\times[y_{1},y_{2}]$ containing $(x,y)$ with the smallest height $y_{2}-y_{1}$ in $\mathcal{O}(\log m)$ time. By swapping all $x$ and $y$ coordinates, we can also find the rectangle with the smallest width $x_{2}-x_{1}$. Note that because no two bridges share an endpoint, those two rectangles are uniquely defined. If they are different, we signal ambiguity. Otherwise this rectangle is the unique lowest bridge. 
\qed
\end{proof}

\section{Decreasing the query time}
\label{section:algorithm2}

In order to reduce the query time we will carefully modify the algorithm from Section~\ref{section:algorithm}. We use a combination of standard tricks (increasing the degree of the tree to $\log^{\epsilon}n$) and a few extension of the observations from the previous section. The high level idea stays the same: we sweep the plane from left to right maintaining a structure describing currently intersected rectangles.  The structure is a full tree of degree $d=\log^{\epsilon}n$ for a sufficiently small $\epsilon$ (to be fixed later) on $n$ leaves corresponding to different $y$ coordinates . Each vertex stores three structures:

\begin{enumerate}
\item left stack storing prefixes of the corresponding segment,
\item right stack storing suffixes of the corresponding segment,
\item block stack storing fragments consisting of a number of whole segments corresponding to a contiguous range of its children.
\end{enumerate}

To process an interval $[y_{1},y_{2}]$ we locate the lowest common ancestor $u$ of the leaves corresponding to $y_{1}$ and $y_{2}$ and split the interval into three parts: a suffix of a segment corresponding to some child of $v_{i}$, a number of full segments corresponding to $v_{i+1},\ldots,v_{j-1}$ (which we call the middle part) and finally a prefix of a segment corresponding to some child $v_{j}$, where $v_{1},v_{2},\ldots,v_{d}$ are the children of $u$.
All stacks will be implemented using technique similar to the one from Lemma~\ref{lemma:version preprocessing}, with the block stack requiring one additional detail. Nevertheless, it also uses the idea of storing a list of all different $x$ coordinates where a new item appears or disappears. For each pair of consecutive $x_{i}<x_{i+1}$ we store a pointer to the corresponding version of the structure, and apply fractional cascading to quickly locate the most up-to-date version given a query $(x,y)$. Now we need a slightly stronger version, though.

\begin{lemma}\label{lemma:fast cascading}
Given a time $t$ we can retrieve the pointers to all current structures at all ancestors of the leaf corresponding to $x$ in total $\mathcal{O}(\frac{\log m}{\log\log m})$ time after a linear time and space preprocessing.
\end{lemma}

\begin{proof}
We apply the {\bf fast} fractional cascading of Shi and  J{\'a}J{\'a}~\cite{Jaja}. It allows a linear time and space preprocessing of a $\log^{\epsilon}m$-degree (where $\epsilon<\frac{1}{5}$) tree with (sorted) lists of elements associated to the vertices so that given a path of length $p$ we can find the successors of a given value in all lists corresponding to its vertices in time $\mathcal{O}(\frac{\log m}{\log\log m}+p)$. As we are working with a full $d$-ary tree, $p\in\mathcal{O}(\frac{\log m}{\log\log m})$ and the claimed running time follows.
\qed
\end{proof}

The implementation of a block stack consists of two parts. For each $i\leq j$ we store a stack of intervals which the middle part corresponds to $v_{i},\ldots,v_{j}$. Note that if we store a pointer for each $i\leq j$, the space usage might be too large. Fortunately, only a linear number of those pointers will be non-null during the whole execution, hence for each vertex we can store a dictionary mapping $(i,j)$ to a pointer. The dictionary can be implemented efficiently using Lemma~\ref{lemma:atomic}. The second part of the implementation is a single word encoding information which stacks are currently nonempty. Observe that the stack of intervals stored for each $i\leq j$ is indeed a stack: we either push a new interval or pop the one which is on the top. Hence its implementation will be the same as in the case of left and right stacks. The only difference is that given $k$ we need to quickly find $i\leq j$ such that the corresponding stack is nonempty, $i\leq k\leq j$ and $j-i$ is smallest. This can be easily done in constant time using the single word encoding all nonempty stacks. To finish the implementation we need to show how to implement all stacks. As in the previous section, we will store their version trees, and a stack is actually a pointer to the current vertex in such tree. Each version tree will be preprocessed using a stronger version of Lemma~\ref{lemma:version preprocessing}. For that we first need to extend Lemma~\ref{lemma:small grid}.

\begin{lemma}\label{lemma:fast small grid}
A collection of  sets of points $S(i)$ on a $n\times n\times n$ grid such that $|S(i)|\leq\log^{3} n$ for all $i$ can be preprocessed in $\mathcal{O}(n+\sum_{i}|S_{i}|)$ time so that given $i$ and $(x,y,z)$ we can retrieve the point corresponding to $\min\{z' : (x',y',z')\in S(i)\wedge x'\leq x\wedge y'\leq y\wedge z'\geq z \}$ in constant time.
\end{lemma}

\begin{proof}
We extend the method used in the proof of Lemma~\ref{lemma:small grid}. We partition $S(i)$ into $yz$, $xz$ and $xy$ blocks of size roughly $m^{3/4}$ by choosing $m^{1/4}$ evenly spaced elements in the sequence of points sorted according to the first, second, and third coordinate, respectively. For each such block we build a smaller structure. Additionally, we create a new set of {\it representatives} by taking all original points and replacing their coordinates with the number of the corresponding $yz$, $xz$, and $xy$ block. Note that the size of this new set is at most $m^{3/4}$. We build a smaller structure for this set of representatives. If $m\leq\log^{1/3}n$, we normalize the coordinates and encode the whole set in a single machine word instead.

To answer a query concerning $(x,y,z)$ we first use the smaller structures built for $yz$, $xz$, and $xy$ blocks. This gives us the answer if its coordinates are close to $x$, $y$ or $z$. Otherwise we use the smaller structure built for the representatives, which gives us the $z$ coordinate of the answer. Having this coordinate, we use the smaller structure built for the corresponding $xy$ block. Overall, the running time is constant, and the total size of all structures is $\mathcal{O}(m)$, as the depth of the recursion is constant.
\qed
\end{proof}

\begin{figure}
\centering
\includegraphics[scale=0.5]{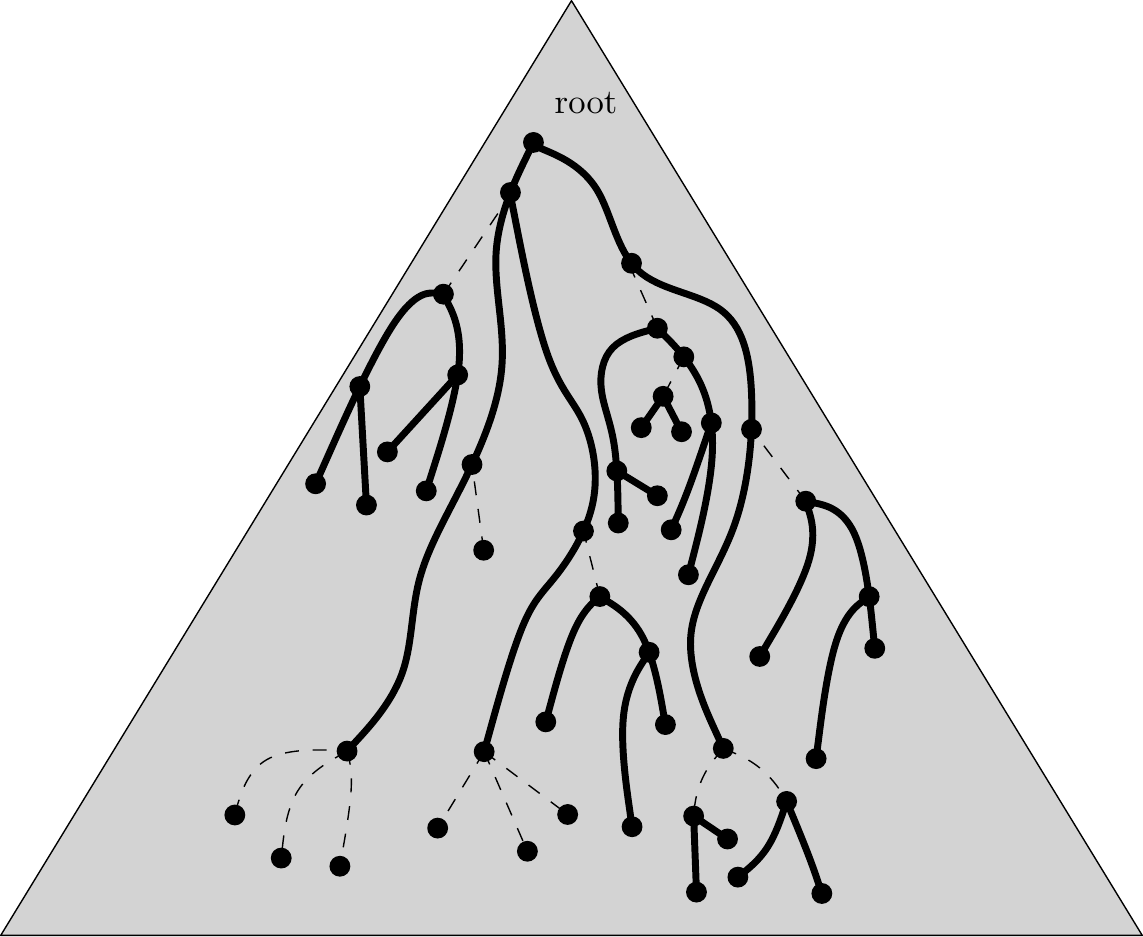}
\caption{Decomposition of a tree into thin fragments.}
\label{figure:thin fragments}
\end{figure}

To extend Lemma~\ref{lemma:version preprocessing}, we need a relaxation of the heavy path decomposition, which we call a {\it thin fragments decomposition}.  First we choose all edges connecting vertices whose subtrees are of size at least $\frac{n}{\log n}$. After removing this top fragment
of the tree we get a collection of smaller trees. For each of them we do the same, namely we choose all edges connecting vertices whose subtrees are of size which is a logarithmic fraction of the whole size of the current tree, remove the resulting top fragments, and repeat the whole procedure until we get single vertices. Each removed fragment is a tree on at most $\log n$ leaves (but possibly many more inner vertices, hence the name), see Figure~\ref{figure:thin fragments}. We define the {\it fragment tree} with vertices corresponding to fragments and edges defined in the natural way, namely we make one vertex a child of another if the root of the former fragment is a child of a vertex belonging to the latter fragment. It is easy to see that its depth is at most $\frac{\log n}{\log\log n}$. {\it Fragment depth} is the depth of the corresponding vertex in the fragment tree.

\begin{lemma}\label{lemma:fast version preprocessing}
A node weighted weighted tree on $n$ vertices with the weights from $\{1,2,\ldots,n\}$ can be preprocessed in linear time and space so that given $v$ and $x$ we can find the vertex corresponding to the successor of $x$ in $\ancestors(v)$ in $\mathcal{O}(\frac{\log n}{\log\log n})$ time.
\end{lemma}

\begin{proof}
For each fragment containing of $\ell$ vertices we construct a sorted list of all weights. We choose $\frac{\ell}{\log^{3}n}$ evenly spaced weights out of them and call the corresponding vertices {\it important} (notice $\log^{3}n$ instead of $\log^{2}n$). For each such important vertex $v_{\alpha\log^{3}n}$ we construct a small set of all pairs
$$\left(-\pre(v_{\alpha\log^{3}n-\Delta}),\post(v_{\alpha\log^{3}n-\Delta}),\weight(v_{\alpha\log^{3}n})\right) \quad\text{ for } \Delta=0,1,\dots,\log^{3}n$$
and apply the preprocessing described in Lemma~\ref{lemma:fast small grid}. Additionally, from each group we select vertices which are minimal in the relation of being an ancestor and call them the {\it representatives}. Note that a single group cannot have more than $\log n$ such representatives, therefore the total number of representatives is just $\mathcal{O}(\frac{n}{\log^{2}n})$.

We will show how to preprocess the tree so that given $v$ and $x$, for each fragment above $v$ we can locate the successor of $x$ among the weights of all important vertices and the weights of all representatives which are above $v$.  We use exactly the same method as in Lemma~\ref{lemma:version preprocessing}, namely for each fragment $f$ we build a binary search tree containing all important vertices on all fragments corresponding to ancestors of $f$ in the fragment tree (including $f$), with a smaller structure stored at each vertex. Similarly, for each representative we build a binary search tree containing all representatives above. Now the size of the helper structure is just $\mathcal{O}(\frac{\log n}{\log\log n})$, though,  hence we can find the successor for each fragment in $\mathcal{O}(\frac{\log n}{\log\log n})$ time after just linear preprocessing. Given the successor, we use the helper structure storing the whole group of the corresponding vertex to retrieve the answer in constant time per fragment above $v$.
\qed
\end{proof}

This gives us all the ingredients necessary to speed up the method from the previous section.

\begin{theorem}
Bridge color problem can be solved in $\mathcal{O}(\frac{\log m}{\log\log m})$ time after a linear time and space preprocessing.
\end{theorem}

\section{Lower bound}
\label{section:lower bound}

It turns out that $\mathcal{O}(\frac{\log m}{\log\log m})$ query time is optimal if we want to keep the size of the structure close to linear. This follows from the results of~\cite{Patrascu}, where a lower bound for the following 2D stabbing problem is shown: preprocess a given collection of $m$ 2D rectilinear rectangles so that we can quickly retrieve (any) rectangle containing a given point. It turns out that if we want to keep the size of the structure $\mathcal{O}(m\log^c m)$, the best possible query time is $\Omega(\frac{\log m}{\log\log m})$, even if we allow randomization. Furthermore, the lower bound is shown through a reduction from the reachability oracle problem, and hence the $x$ and $y$ projections of any two rectangles in the collection are either disjoint or contained in each other. A closer inspection of the proof shows that the queries can be assumed to be chosen so that there is at most one rectangle containing the point. It follows that we can encode the instance as a binary method dispatching problem, and the claimed lower bound follows.

\section{Conclusions}

We presented a time-and-space optimal solution for the binary method dispatching problem. Two questions remain:

\begin{enumerate}
\item is it possible to make the structure dynamic?
\item can we achieve linear preprocessing and logarithmic query in the pointer machine model? Note that the construction of Alstrup et al.~\cite{Alstrup} makes a heavy use of word RAM specific structures, such as the van Emde Boas trees.
\end{enumerate}

\bibliographystyle{abbrv}
\bibliography{biblio}

\end{document}